\documentclass[11pt,letterpaper]{article}

\usepackage[margin=1in]{geometry}
\usepackage[english]{babel}
\usepackage[utf8x]{inputenc}
\usepackage{amsmath}
\usepackage{amsfonts}
\usepackage{soul}
\usepackage{amssymb}
\usepackage{graphicx}
\usepackage{rotating}
\usepackage{amsbsy}
\usepackage{subfig}
\usepackage{graphicx,  ucs}
\usepackage{float}
\usepackage{tikz}
\usepackage{algorithm}
\usepackage[noend]{algpseudocode}
\usepackage{listings}
\usepackage{amsthm}
\usepackage{enumitem}
\usepackage{hyperref}
\usepackage[compact]{titlesec}

\def\compactify{\itemsep=0pt \topsep=0pt \partopsep=0pt \parsep=0pt}
\let\latexusecounter=\usecounter

\newcommand{\Domain}{\Omega}
\newcommand{\Support}{\mathcal{X}}
\newcommand{\Set}{S_{\mathcal{X}}}

\newcommand{\tl}{\textlatin}

\newcommand{\argmax}{\text{argmax}}

\def\poly{\mathrm{poly}}
\def\eps{\varepsilon}

\def\Prob{\mathbb{P}}

\def\Alg{\textsc{Alg}}

\def\Cond{\textsc{Cond}}
\def\EP{\textsc{EP}}
\def\SE{\textsc{SE}}
\def\DV{\textsc{DV}}
\def\MAX{\textsc{Max}}
\def\ARGMAX{\textsc{ArgMax}}
\def\SUM{\textsc{Sum}}
\def\WCOND{\textsc{WCond}}
\def\DES{\textsc{DES}}

\def\SPACE{\textsc{SPACE}}

\newcommand{\reals}{\mathbb{R}}
\newcommand{\nats}{\mathbb{N}}

\def\abs#1{\left|#1\right|}

\newtheorem{theorem}{Theorem}

\newtheorem{corollary}{Corollary}
\newtheorem{lemma}{Lemma}

\definecolor{vergreen}{RGB}{0,85,2}
\definecolor{myvergreen}{RGB}{0,140,3}
\definecolor{provorange}{RGB}{85,34,0}
\definecolor{inputblue}{RGB}{5,13,111}
\definecolor{noapred}{RGB}{116,3,3}
\definecolor{classesblue}{RGB}{9,49,146}

\definecolor{secinhead}{RGB}{249,196,95}

\definecolor{lgray}{gray}{0.8}

\begin{document}

\sethlcolor{lgray}

\begin{titlepage}
\title{Faster Sublinear Algorithms using Conditional Sampling}
\date{}
\author{
  Themistoklis Gouleakis \\ EECS, MIT \\ \href{mailto:tgoule@mit.edu}{tgoule@mit.edu}	
  \and Christos Tzamos \\ EECS, MIT \\ \href{mailto:tzamos@mit.edu}{tzamos@mit.edu}	
	\and Manolis Zampetakis \\ EECS, MIT \\ \href{mailto:mzampet@mit.edu}{mzampet@mit.edu}	
	}
\clearpage
\maketitle
\thispagestyle{empty}
\begin{abstract}
A conditional sampling oracle for a probability distribution $D$ returns samples from the conditional distribution of $D$ restricted to a specified subset of the domain. A recent line of work \cite{chakraborty2013power, canonne2014testing} has shown that having access to such a conditional sampling oracle requires only polylogarithmic or even constant number of samples to solve distribution testing problems like identity and uniformity. This significantly improves over the standard sampling model where polynomially many samples are necessary.

Inspired by these results, we introduce a computational model based on conditional sampling to develop sublinear algorithms with exponentially faster runtimes compared to standard sublinear algorithms. We focus on geometric optimization problems over points in high dimensional Euclidean space. Access to these points is provided via a conditional sampling oracle that takes as input a succinct representation of a subset of the domain and outputs a uniformly random point in that subset.
We study two well studied problems: k-means clustering and estimating the weight of the minimum spanning tree. In contrast to prior algorithms for the classic model, our algorithms have time, space and sample complexity that is polynomial in the dimension and polylogarithmic in the number of points.

Finally, we comment on the applicability of the model and compare with existing ones like streaming, parallel and distributed computational models.
\end{abstract}
\end{titlepage}
  
  \makeatletter{}\section{Introduction}
\label{s:intro}

Consider a scenario where you are given a dataset of input points $\Support$, from some domain $\Domain$, stored in a random access memory and you want to estimate the number of distinct elements of this (multi-)set. One obvious way 
to do so is to iterate over all elements and use a hash table to find duplicates. Although simple, this solution becomes unattractive if the input is huge and it is too expensive to even parse it. In such cases, one natural goal is 
to get a good estimate of this number instead of computing it exactly. One way to do that is to pick some random numbers from $\Support$ and estimate, based on those, the total number of distinct elements in the set. This is 
equivalent to getting samples from a probability distribution where the probability of each element is proportional to the number of times it appears in $\Support$. In the context of probability distributions, this is a well 
understood problem, called \textit{support estimation}, and tight bounds are known for its sampling complexity. More specifically, in \cite{valiantV11}, it is shown that the number of samples needed is $\Theta(n / \log n)$ which,
although sublinear, still has a huge dependence on the input size $n = |\Support|$.

In several situations, more flexible access to the dataset might be possible, e.g. when data are stored in a database, which can significantly reduce the number of queries needed to perform support estimation or other tasks. One 
recent model, called \textit{conditional sampling}, introduced by~\cite{chakraborty2013power, canonne2014testing} for distribution testing, describes such a possibility. In that model, there is an underlying distribution $D$, and a 
conditional sampling oracle takes as input a subset $S$ of the domain and produces a sample from $D$ conditioned to a set $S$. \cite{chakraborty2013power} and \cite{canonne2014testing} study several problems in distribution testing 
obtaining surprising results: Using \textit{conditional queries} it is possible to avoid polynomial lower bounds that exist for the sampling complexity in the standard distribution testing framework and get testers with only 
polylogarithmic or even constant query requirements. In follow up work, Acharya, Cannone and Kamath \cite{acharyaCK15} consider the support estimation problem we described above and prove that the support estimation problem can be 
solved using only $O(\poly \log \log n)$ conditional samples.
This is a doubly exponentially better guarantee compared to the optimal classical algorithm which requires $\Theta(n / \log n)$ samples \cite{valiantV11}. 

Inspired by the power of these results, we introduce a computational model based on conditional sampling where the dataset is provided as a distribution and algorithms have access to a conditional sampling oracle that returns 
datapoints at random from a specified set. More precisely, an algorithm is given access to an oracle $\Cond(C)$ that takes as input a function $C: \Domain \rightarrow \{0,1\}$ and returns a tuple $(i, x_i)$  with $C(x_i) = 1$ with 
$i$ chosen uniformly at random from the subset $\{ j \in [n] \mid C(x_j) = 1 \}$. If no such tuple exists the oracle returns $\bot$. The function $C$ is represented as a binary circuit. We assume that queries to the conditional 
sampling oracle $\Cond$ take time linear in the circuit size. Equivalently, we could assume constant time, as we are already paying linear cost in the size of the circuit to construct it.

Most works in conditional sampling, measure the performance of algorithms only by their query complexity. The work of \cite{canonne2014testing} considers the description complexity of the query set $S$ by examining restricted conditional queries that either specify pairs of points or intervals. However, in many cases, such simple queries might not be sufficient to obtain efficient algorithms. We use the circuit size of a set’s description as a measure of simplicity to allow for richer queries which is naturally counted towards the runtime of our algorithms.

Except from its theoretical interest, it is practically useful to consider algorithms that perform well in the conditional sampling model. This is because efficient algorithms for the conditional sampling model can directly be 
implemented in a number of different computational models that arise when we have to deal with huge amount of data. In particular, let's assume that we have an algorithm $A$ that solves a problem $P$ using $q$ conditional queries 
where the description of the sets used has size $s$ and the additional computational time needed is $r$. 

\paragraph{Parallel Computing: } We notice that the computation of one conditional sample can be very easily parallelized because it suffices to assign to each processor a part of the 
    input and send to each of them the description of the circuit. Each processor can compute which of its points satisfy the circuit and pick one at random among them. Then, we can select as output the sample of one processor chosen
    at random. The probability of choosing one processor in this phase is proportional to the number of points in the input assigned to this 
    processor that belong to the conditioning set. This way we can implement in just a few steps a conditional sampling query. If the input is divided evenly among $m$ processors the load on each of them is $n/m$. Combining the 
    answers can be done in $\log m$ steps and therefore, the running time of $A$ in the parallel computation model is $O(q \cdot s \cdot (n / m + \log m) + r)$ which gives a non-trivial parallelization of the problem $P$. Except from
    the running time, one  important issue that can decrease the performance of a parallel algorithm is the communication that is needed among the processors as described in the work of Afrati et. al. \cite{afratiBSHSU12}. This 
    communication cost can be bounded by the size $s$ of the circuit at each round plus the communication for the partition of the input that happens only once in the beginning.
\paragraph{Streaming Algorithms: } The implementation of a conditional query in the streaming model where we want to minimize the number of passes of the input is pretty straightforward. With one pass of the input we
    can select one point uniformly at random from the points that belong to the conditioning set using standard streaming techniques. The space that we need for each of these passes is just 
    $s$ and we need $q$ passes of the input. 
\paragraph{Distributed Computing: } The implementation of a conditional query in the distributed computational model can follow the same ideas as in the parallel computational model.
 
The surprising result in all the above cases is that one has to deal with transferring appropriately the conditional sampling model to the wanted computational model and then we can get high performance
algorithms once $q$, $s$ and $r$ are small. In this work we design algorithms that achieve all these quantities to be only polylogarithmic in the size of the input, which leads to very efficient algorithms in all the above models. 

\makeatletter{}\subsection{Previous Work on Sublinear Algorithms}

 We consider two very well studied combinatorial problems: $k$-means clustering and minimum spanning tree. For these problems we know the following about the sublinear algorithms in the classical setting.

\subsubsection{$k$-means Clustering}
  
  Sublinear algorithms for $k$-median and $k$-means clustering first studied by Indyk \cite{indyk99}. In this work, given a set of $n$ points from a metric space, an algorithm is given that computes a set of size $O(k)$ that 
approximates the cost of the optimal clustering within a constant factor and runs in time $O(n k)$. Notice that the algorithm is sublinear as the input contains all the pairwise distances between the points which have total size 
$O(n^2)$. 

In followup work, Mettu and Plaxton \cite{mettuP04} gave a randomized constant approximation algorithm for the $k$-median problem with running time 
$O(n ( k + \log n) )$ subject to the constraint $R \le 2^{O(n / \log (n / k))}$, where $R$ denotes the ratio between the maximum and the minimum distance between any pair of distinct points in the metric 
space. Also Meyerson et. al. \cite{meyersonOP04} presented a sublinear algorithm for the $k$-median problem with running time $O((k^2 / \eps) \log (k / \eps))$ under the assumption that each cluster has
size $\Omega(n \eps / k)$.

  In a different line of work Mishra, Oblinger and Pitt \cite{mishraOP01} and later Czumaj and Sohler \cite{czumajS07} assume that the diameter $\Delta$ of the set of points is bounded and known. The 
running time of the algorithm by Mishra et. al. \cite{mishraOP01} is only logarithmic in the input size $n$ but is polynomial in $\Delta$. Their algorithm is very simple since it just picks uniformly at 
random a subset of points and solves the clustering problem on that subset. Following similar ideas, Czumaj and Sohler \cite{czumajS07} gave a tighter analysis of the same algorithm proving that the running time depends only on the diameter $\Delta$ and is independent of $n$. The dependence on $\Delta$ is still polynomial in this work. The guarantee in both these works is a constant multiplicative approximation algorithm with an additional additive error term.

\subsubsection{Minimum Spanning Tree in Euclidean metric space}

There is a large body of work on sublinear algorithms for the minimum spanning tree. In \cite{indyk99}, given $n$ points in a metric space $\Domain$ an algorithm is provided that  outputs a spanning tree in time $\tilde{O}(n / \delta)$ achieving a $(1/2 - \delta)$-approximation to the optimum.
When considering only the task of estimating of the weight of the optimal spanning tree, Czumaj and Sohler 
\cite{CS04} provided an algorithm that gets an $(1 + \eps)$-approximation. The running time of this algorithm is $\tilde{O}(n \cdot \poly (1/\eps))$.

  To achieve better guarantees several assumptions could be made. One first assumption is that we are given a graph that has bounded average degree $deg$ and the weights of the edges are also bounded by
$W$. For this case, the work of Chazelle et. al. \cite{CRT05} provides a sublinear
algorithm with running time $\tilde{O}(deg \cdot W \cdot 1/\eps^2)$ that returns the weight of the minimum spanning tree with relative error $\eps$. Although the algorithm completely gets rid of the
dependence in the number of points $n$ it depends polynomially in the maximum weight $W$. Also in very dense graphs $deg$ is polynomial in $n$ and therefore we also have a polynomial dependence on $n$.

  Finally, another assumption that we could make is that the points belong to the $d$-dimensional Euclidean space. For this case, the work of Czumaj et. al. \cite{CEFMNRS05} provide an $(1 + \eps)$-approximation algorithm that requires time $\tilde{O}(\sqrt{n} \cdot (1/\eps)^d)$. Note that in this case the size of the input is $O(n)$ and not $O(n^2)$ since given the coordinates of the $n$ 
points we can calculate any distance. Therefore, the algorithms described before that get running time $O(n)$ are not sublinear anymore. Although Czumaj et. al. \cite{CEFMNRS05} manage to achieve a 
sublinear algorithm in this case they cannot escape from the polynomial dependence on $n$. Additionally, their algorithm has exponential dependence on the dimension of the Euclidean space.

\subsection{Our Contribution}

  The main result of our work is that on the conditional sampling framework we can get exponentially faster sublinear algorithms compared to the sublinear algorithms in the classical framework.

  We first provide some basic building blocks -- useful primitives for the design of algorithms. These building blocks are:
\vspace{-2mm}
\begin{enumerate}[label=\alph*.]
  \setlength\itemsep{-1mm}
  \item Compute the size of a set given its description, Section \ref{s:primSE}.  \item Compute the maximum of the weights of the points of a set given the description of the set and the description of the weights, Section \ref{s:primMAX}.  \item Compute the sum of the weights of the points of a set given the description of the set and the description of the weights, Section \ref{s:primSUM}.  \item Get a \textit{weighted conditional sample} from the input set of points given the description of the weights, Section \ref{s:primWCOND}.  \item Get an \textit{$\ell_0$-sample} given the description of labels to the points Section \ref{s:primDES}.\end{enumerate}

  For all these primitives, we give algorithms that run in time polylogarithmic in the domain size and the value range of the weights. We achieve this by querying the conditional sampling oracle with random subsets produced by 
appropriately chosen distribution on the domain. Intuitively, this helps to estimate the density of the input points on different parts of the domain. One important issue of conditioning on random sets in that the description 
complexity of the set can be almost linear on the domain size. To overcome this difficulty we replace the random sets with pseudorandom ones based on Nisan's pseudorandom generator \cite{Nisan90}. The implementation of these 
primitives is of independent interest and especially the fourth one since it shows that the \textit{weighted conditional sample}, which is similar to sampling models that have been used in the literature \cite{acharyaCK15b}, can be 
simulated by the conditional sampling model with only a polylogarithmic overhead in the query complexity and the running time. 

  After describing and analyzing these basic primitives, we use them to design fast sublinear algorithms for the $k$-means clustering and the minimum spanning tree.

\subsubsection{$k$-means Clustering}

  Departing from the works of Mishra, Oblinger and Pitt \cite{mishraOP01} and Czumaj and Sohler \cite{czumajS07} where the algorithms start by choosing a uniform random subset, we start by choosing a 
random subset based on \textit{weighted sampling} instead of uniform. In the classical computational model we need at least linear time to get one conditional sample and thus it is not possible
to use the power of weighted sampling to get sublinear time algorithms for the $k$-means problem. But when we are working in the conditional sampling model, then the adaptive sampling can be implemented in
polylogarithmic time and queries. This enables us to use all the known literature about the ways to get efficient algorithms using conditional sampling \cite{arthurV07}.
  Quantitatively the advantage from the use of the weighted sampling is that we can get sublinear algorithms with running times $O(\poly( \log \Delta, \log n))$ where $\Delta$ is the diameter of the 
metric space and $n$ the number of points on the input. This is exponentially better from Indyk \cite{indyk99} in terms of $n$ and exponentially better from Czumaj and Sohler \cite{czumajS07} in terms of 
$\Delta$. This shows the huge advantage that one can get from the ability to use or implement conditional sampling queries. We develop and analyze these ideas in detail in Section \ref{s:kmeans}.

\subsubsection{Minimum Spanning Tree in Euclidean metric space}

  Based on the series of works on sublinear algorithms for minimum spanning trees, we develop algorithms that exploit the power of conditional sampling and achieve polylogarithmic time with respect to the number of input points $n$ 
and only polynomial with respect to the dimension of the Euclidean space. This is very surprising since in the classical model it seems to exist a polynomial barrier that we cannot escape from. Compared to the algorithm by Czumaj et.
al. \cite{CEFMNRS05}, we get running time $O(\poly(d, \log n, 1/\eps))$ which is exponential improvement with respect to both the parameters $n$ and $d$.
  
  We present our algorithm at Section \ref{s:mst}. From a technical point of view, we use a gridding technique similar to \cite{CEFMNRS05} but prove that using a random grid can significantly reduce the runtime of the algorithm as we
avoid tricky configurations that can happen in worst case.

  \makeatletter{}\section{Model and Preliminaries}
\label{s:prelim}

\paragraph{Notation} For $m \in \nats$ we denote the set $\{1,\cdots,m\}$ by $[m]$. We use $\tilde O(N)$ to denote $O(N \log^{O(1)} N)$ algorithms.

  Given a function $f$ that takes values over the rationals we use $C_f$ to denote the binary circuit that takes as input the binary representation of the input $x$ of $f$ and outputs the binary 
representation of the output $f(x)$. If the input or the output are rational numbers then the representation is the pair $(numerator, denominator)$.

Suppose we are given an input $\vec x = (x_1, x_2 , \cdots , x_n)$ of length $n$, where every $x_i$ belongs in some set $\Domain$. In this work, we will fix $\Domain = [\mathcal{D}]^d$ for some $\mathcal{D} = n^{O(1)}$ to be the discretized $d$-dimensional Euclidean space.
Our goal is to compute the value of a symmetric function $f : \Domain^n \rightarrow \mathbb{R}_+$ in input $\vec x \in \Domain^n$. We assume that all $x_i$ are distinct and define $\Support \subseteq \Domain$ as the set $\Support = \{x_i : i \in [n]\}$. Since we consider symmetric functions $f$, it is convenient to extend the definition of $f$ to sets $f(\Support) = f(x)$.

A randomized algorithm that estimates the value $f(x)$ is called \textit{sublinear} if and only if its running time is $o(n)$. We are interested in \textit{additive} or \textit{multiplicative approximation}. A sublinear algorithm \textsc{Alg} for computing $f$ has $(\eps, \delta)$-additive 
approximation if and only if
\[ \Prob \left[ | \Alg(x) - f(x) | \ge \eps \right] \le \delta \]
and has $(\eps, \delta)$-multiplicative approximation if and only if 
\[ \Prob \left[ (1 - \eps) f(x) \le \Alg(x) \le (1 + \eps) f(x) \right] \le \delta. \]
We usually refer to $(\eps, \delta)$-approximation and is clear from the context if we refer to the additive or the
multiplicative one.

\subsection{Conditional Sampling as Computational Model}
\label{s:prelCond}

The standard sublinear model assumes that the input is stored in a random access memory that has no further structure. Since $f$ is symmetric in the input points, the only reasonable operation is to \textit{uniformly sample} points from the input. Equivalently, the input can be provided by an oracle \textsc{Sub} that returns a tuple $(i, x_i)$ where $i$ is chosen uniformly at random from the set $[n] = \{1, \dots, n\}$.

When the input has additional structure (i.e. points stored in a database), more complex queries can be performed.
The \emph{conditional sampling model} allows such queries of small description complexity to be performed.
In particular, the algorithm is given access to an oracle $\Cond(C)$ that takes as input a function $C: \Domain \rightarrow \{0,1\}$ and returns a tuple $(i, x_i)$ with $C(x_i) = 1$ with $i$ chosen uniformly at random from the subset $\{ j \in [n] \mid C(x_j) = 1 \}$. If no such tuple exists the oracle returns $\bot$. The function $C$ is represented as a sublinear sized binary circuit. All the results presented in this paper use polylogarithmic circuit sizes.

We assume that queries to the conditional sampling oracle $\Cond$ take time linear in the circuit size. Equivalently, we could assume constant time, as we are already paying linear cost in the size of the circuit to construct it.

\subsection{$k$-means Clustering}

  Let $d( \cdot, \cdot )$ be distance metric function $d : \Domain \times \Domain \rightarrow \reals$, i.e. $d(x, y)$ represents the distance between $x$ and $y$. 
  Given a set of 
\textit{centers} $P$ we define the distance of a point $x$ from $P$ to be $d(x, P) = \min_{c \in P} d(x, c)$. Now given a set of $n$ input points $\Support \subseteq \Domain$ and a set of centers 
$P \subseteq \Omega$ we define the cost of $P$ for $\Support$ to be $d(\Support, P) = \sum_{x \in \Support} d(x, P)$. The $k$-means problem is the problem of minimizing the \textit{squared cost} 
$d^2(\Support, P) = \sum_{x \in \Support} d^2(x, P)$ over the choice of centers $P$ subject to the constraint $|P| = k$. We assume that the diameter of the metric space is 
$\Delta = \max_{x, y \in \Support} d(x, y)$.

\subsection{Minimum spanning tree in Euclidean space}

Given a set of points $\Support$ in $d$ dimensions, the minimum spanning tree problem in Euclidean space ask to compute the a spanning tree $T$ on the points minimizing the sum of weights of the edges. The weight of an edge between two points is equal to their Euclidean distance. 

We will focus on a simpler variant of the problem which is to compute the weight of the best possible spanning tree, i.e. estimate the quantity
$\min_{\text{tree } T} \sum_{(x,x') \in T} \|x - x'\|_2$.

  \makeatletter{}\section{Basic Primitives}
\label{s:primi}
  
In this section, we describe some primitive operations that can be efficiently implemented in this model. We will use these primitives as black boxes in the algorithms for the combinatorial problems we consider. We make this separation as these primitives are commonly used building blocks and will make the presentation of our algorithms cleaner.

A lot of the algorithmic primitives are based on constructing random subsets of the domain and querying the random oracle $\Cond$ with a description of this set. A barrier is that such subsets have description complexity that is linear in the domain size. For this reason, we will use a pseudorandom set whose description is polylogarithmic in the domain size. The main tool to do this is Nisan's pseudorandom generator \cite{Nisan90} which produces pseudorandom numbers that appear as perfectly random to algorithms running in polylogarithmic time.

\begin{theorem}[Nisan's Pseudorandom Generator \cite{Nisan90}]
\label{th:nisan}
Let $U_N$ and $U_{\ell}$ denote uniformly random binary sequences of length $N$ and $\ell$ respectively.
There exists a map $G : \{0, 1\}^{\ell} \rightarrow \{0, 1\}^{N}$ such that for any algorithm $A : \{0, 1\}^{N} \rightarrow \{0, 1\}$, with $A \in \SPACE(S)$, where $S = S(N)$,
it holds that 
  \[ \abs{\Prob(A(U_{N}) = 1) - \Prob(A(G(U_{\ell})) = 1)} \le 2^{-S} \]
for $\ell = \Theta(S \log N)$. 
\end{theorem}

Nisan's pseudorandom generator is a simple recursive process that starts with $\Theta(S \log N)$ random bits and generates a sequence of $N$ bits. The sequence is generated in blocks of size $S$ and every block can be computed given the seed of size $\Theta(S \log N)$ using $O(\log N)$ multiplications on $S$ bit numbers. The overall time and space complexity to compute the $k$-th block of $S$ bits is $\tilde O(S \log N)$ and there exists a circuit of size $\tilde O(S \log N)$ that performs this computation.

Using Nisan's theorem, we can easily obtain pseudorandom sets for conditional sampling. We are interested in random sets where every element appears with probability $g(x)$ for some given function $g$. 
\begin{corollary}[Pseudorandom Subset Construction]
\label{cor:prsc}
    Let $R$ be a random set, described by a circuit $C_R$, that is generated by independently adding each element $x \in \Domain$ with probability $g(x)$, where $g$ is described by a circuit $C_g$. For any $\delta \ge |\Domain|^{-1}$, there exists a random set $R'$ described by a $\tilde O(|C_g| + \log|\Domain| \log(1/\delta) )$-sized circuit $C_{R'}$ such that
  \[ \abs{\Prob(\Cond(C \wedge C_R) = x) - \Prob(\Cond(C \wedge C_{R'}) = x)} \le \delta\]
for all circuits $C$ and elements $x \in \Domain$.
\end{corollary}

\begin{proof}
  The corollary is an application of Nisan's pseudorandom generator for conditional sampling. A simple linear time algorithm that performs conditional sampling based on a random set $R$ as follows. We keep two variables, $cnt_{matched}$ the number of elements that pass the criteria of selection which is initialized at value 0, and the selected element. For every element $x$ in the domain $\Domain$ in order, we perform the following:
\begin{enumerate}
  \item Draw $k$ random bits $\vec b \in \{0,1\}^k$ and check whether the number $b \cdot 2^{-k} > g(x)$.
  \item If yes, skip $x$ and continue in the next element.
  \item Otherwise if $C(x) = 1$, increment $cnt_{matched}$ and with probability $cnt_{matched}^{-1}$ change the selected element to $x$.
\end{enumerate}

Note that here, we have truncated the probabilities $g(x)$ to $2^{-k}$ accuracy, so the random set $\bar R$ used is slightly different than $R$. However, picking $k = O(\log(|\Domain|/\delta))$, we have that
\[ \abs{\Prob(\Cond(C \wedge C_R) = x) - \Prob(\Cond(C \wedge C_{\bar R}) = x)} \le \frac {\delta} 2\]
for all circuits $C$ and elements $x \in \Domain$.

To prove the statement, we will use Nisan's pseudorandom generator to generate the sequence of bits for the algorithm. The algorithm requires only memory to store the two variables which is equal to $\Theta( \log{|\Domain|}) \le k$. Moreover, the total number of random bits used is $k |\Domain|$ and thus by Nisan's pseudorandom generator we can create a sequence of random bits $R'$ based on a seed of size $O(\log(k |\Domain|)$ and give them to the algorithm. This sequence can be computed in blocks of size $k$ using a circuit $C^r$ of size  $\tilde O(\log(k |\Domain|) \log(1/\delta) ) = \tilde O(\log(|\Domain|) \log(1/\delta) )$. We align blocks of bits with points $x \in \Domain$ and thus the circuit $C^r$ gives for input $x$ the $k$ bits needed in the first step of the above algorithm. This implies that the circuit $C_{R'}$ that takes the output of $C^r$ and compares them with $C_g$ satisfies:
\[ \abs{\Prob(\Cond(C \wedge C_{\bar R}) = x) - \Prob(\Cond(C \wedge C_{R'}) = x)} \le \frac {\delta} 2\]
for all circuits $C$ and elements $x \in \Domain$. By triangular inequality, we get the desired error probability with respect to the circuit $C_R$.

The total size of the circuit $C_{R'}$ is $\tilde O(|C_g| + \log|\Domain| \log(1/\delta) )$ which completes the proof.
\end{proof}

\makeatletter{}\subsection{Point in Set and Support Estimation}
\label{s:primSE}

\subsubsection{Point in Set}
 
  The \textit{point in set} function takes a set $S \subseteq \Domain$ given as a circuit $C$ and returns one point $x \in S$ or $\bot$ if there is no such point in the set of input points, i.e. $\Support \cap S = \varnothing$. The notation that we use for this function is $\textsc{EP}( \cdot )$ and takes as input the description $C$ of $S$. Obviously the way to implement this function in the conditional 
sampling model is trivial. Since the point in set returns any point in the described set $S$ a random point also suffices. Therefore we just call the oracle $\Cond(C)$ and we return this as a result to $\EP(C)$.

We can test whether there is a unique point in a set by setting $x^* = \EP(C)$ and querying $\EP(C \wedge \mathbb{I}_{x \neq x^*})$. Similarly, if the set has $k$ points, we can get all points in the set in time $O(|C| k + k^2)$ by querying $k$ times, excluding every time the points that have been found.

\subsubsection{Support Estimation}

  The \textit{support estimation} function takes as input a description $C$ of a set $S \subseteq \Domain$ and outputs an estimation for the size of the set $\Set = S \cap \Support$. We call this function $\SE(C)$. \\
  
  The first step is to define a random subset $R \subseteq \Domain$ by independently adding every element $x \in \Domain$ with probability  $\frac{1}{\alpha}$
for some integer parameter $\alpha$ that corresponds to a guess of the support size. Let $C_R$ be the description of $R$. We will later use Corollary~\ref{cor:prsc} to show that an approximate version of $C_R$ can be efficiently 
constructed. We then use the Point-In-Set primitive and we query $\EP(C \wedge C_R)$. This tests whether $S_{\Support} \cap R \stackrel{?}{=} \varnothing$ which happens with probability
\[ \Prob[ S_{\Support} \cap R = \varnothing ] = \Prob[ ( s_1  \notin R ) \wedge ( s_2 \notin R ) \wedge \cdots \wedge ( s_k \notin R ) ] = \left( 1 - \frac{1}{\alpha} \right)^k. \]

Using this query, we can distinguish whether $|S_{\Support}| \le (1-\eps) \alpha$ or $|S_{\Support}| \ge (1+\eps) \alpha$. The probabilities of these events are $P_1 \ge \left( 1 - \frac{1}{\alpha} \right)^{\alpha (1+\eps)}$ and $P_2 \le \left( 1 - \frac{1}{\alpha} \right)^{\alpha (1-\eps)}$ respectively. The total variation distance in the two cases is 
  \begin{equation*}
  \label{eq:primSE1}
    P_1 - P_2 = P_1 \left( 1 - \frac{P_2}{P_1} \right) \ge \left( 1 - \frac{1}{\alpha} \right)^{\alpha (1-\eps)} \left( 1 - \left( 1 - \frac{1}{\alpha} \right)^{\alpha \cdot 2 \eps} \right) \ge  \left( \frac{1}{4} \right)^{1 - \eps} \left( 1 - e^{-2 \eps} \right) \ge \frac{\eps}{2 e}
  \end{equation*}
where for the second to last inequality we assumed $\alpha \ge 2$\footnote{The case $a=1$ can be trivially handled by listing few points from $S_{\Support}$.}. 

We can therefore conclude that for we can distinguish with probability $\delta$ between $|S| \le (1 - \eps) \alpha$ and $|S| \ge (1 + \eps) \alpha$ using $O(\log \frac 1 {\delta} /\eps^2)$ queries of the form $\EP(C \wedge C_R)$. Binary searching over possible
$\alpha$'s, we can compute an $(1 \pm \eps)$ approximation of the support size by repeating $O(\log n)$ times, as there are $n$ possible values for $\alpha$. A more efficient estimator, since we care about multiplicative approximation, only considers values for $\alpha$ of the form $(1 + \eps)^i$. There are 
are $\log_{1 + \eps} n = O(\frac 1 {\eps} \log n)$ possible such values, so doing a binary search over them takes $O(\log \frac 1 {\eps} + \log \log n)$ iterations.  
Thus, overall, the total number of queries is $\tilde O( \log \frac 1 {\delta} \log \log n  /\eps^2)$. 

To efficiently implement each query, we produce a circuit $C_{R'}$ using Corollary~\ref{cor:prsc} with parameter $\delta'$ for error and a constant function $g(x) = 1/2$. The only change is that at every comparison the probabilities $P_1$ and $P_2$ are accurate to within $\delta'$. Choosing $\delta' = \frac {\eps} {|\Omega|}$ implies that $|P_1 - P_2|$ is still $\Omega(\eps)$ and thus the same analysis goes through. The circuit $C \wedge C_{R'}$ has size $\tilde O(|C| + \log^2(|\Domain|) + \log(|\Domain|) \log(1/\eps) )$ which implies that the total runtime for $\tilde O( \log \frac 1 {\delta} \log \log n  /\eps^2)$ queries is $\tilde O \left( (|C| + \log^2(|\Domain|) ) \log \frac 1 {\delta} /\eps^2 \right)$ as $n = O(|\Omega|)$.

Using our conditional sampling oracle, we are able to obtain the following lemma:
\begin{lemma}
  \label{l:primSE}
    There exists a procedure $\SE(C)$ that takes as input the description $C$ of a set $S$ and computes an $(\eps, \delta)$-multiplicative approximation of the size of $S$ using $\tilde O(\log \log n \log(1 / \delta) / \eps^2 )$
  conditional samples in time $\tilde O\left( (|C| + \log^2 |\Domain|) \log(1 / \delta) / \eps^2 \right)$.
\end{lemma}

\subsubsection{Distinct Values}

  One function that is highly related to support estimation is the \textit{distinct values} function, which we denote by $\DV$. The input of this function is a description $C$ of a set $S$ together with a function $f : \Domain \rightarrow [M]$ described by a circuit $C_f$. The output
of $\DV$ is the total number of distinct values taken by $f$ on the subset $\Set = S \cap \Support$, i.e.
\[ \DV(C, C_f) = |\{ f(x) \mid x \in \Set\}| \]

To implement the distinct values function we perform support estimation on the range $[M]$ of the function $f$. This is done as before by computing random sets $R \subseteq [M]$ and performing the queries $\EP( C \wedge (C_R \circ C_f) )$, where $C_R$ is the circuit description of the set $R$ and the circuit $C_R \circ C_f$ takes value $1$ on input $x$ if $f(x) \in R$ and $0$ otherwise.

\begin{lemma}
  \label{l:primDV}
    There exists a procedure $\DV(C, f)$ that takes as input the description $C$ of a set $S$ and a function $f: \Domain \rightarrow [M]$ given as a circuit $C_f$ and computes an $(\eps, \delta)$-multiplicative approximation to the number of distinct values of $f$ on the set $S \cap \Support$  in time $\tilde O(( |C| + |C_f| + \log^2 M ) \log(1 / \delta) / \eps^2 )$. 
  The number of conditional samples used is  $\tilde O( \log \log n \log(1 / \delta) / \eps^2 )$.
\end{lemma}
 
\makeatletter{}\subsection{Point of Maximum Weight}
\label{s:primMAX}

  The \textit{point of maximum weight} function takes as input a description $C$ of a set $S$ together with a function $f : \Domain \rightarrow [M]$ given by a circuit $C_f$. Let $\Set = S \cap \Support$. The output of the function is the value $\max_{x \in \Set} f(x)$. We call this function 
$\MAX(C, C_f)$. Sometimes we are interested also in finding a point where this maximum is achieved, i.e. $\arg \max_{x \in \Set} f(x)$ which we call $\ARGMAX(C, C_f)$. \\

This is simple to implement by binary search for the value $\MAX(C, C_f)$. At every step, we make a guess $m$ for the answer and test whether there exists a point in the set $\Set \cap \{f(x) \ge m\}$. This requires $\log M$ queries and the runtime is $O( (|C| + |C_f|) \log M )$.

\begin{lemma}
  \label{l:primMAXwM}
    There exists a procedure $\MAX(C, C_f)$ that takes as input the description $C$ of a set $S$ and a function $f: \Domain \rightarrow [M]$ given as a circuit $C_f$ and computes the value $\max_{x \in \Set} f(x)$ using $O( \log M )$ conditional samples in time $O( ( |C| + |C_f| ) \log M )$.
\end{lemma}

An alternative algorithm solves this task in $\tilde O(\log n \log(1/\delta) )$ queries with probability $\delta$. The algorithm starts with the lowest possible value $m$ for the answer, i.e. $m = 1$. At every step, it asks the $\Cond$ oracle for a random point with $f(x) > m$. If such a point $x^*$ exists, the algorithm updates the estimate by setting $m = f(x^*)$. Otherwise, if no such point exists, the guessed value of $m$ is optimal. It is easy to see that every step, with probability $1/2$, half of the points are discarded. 
Repeating $\log( \log n/\delta)$ times we get that the points are halfed with probability $1 - \delta/\log n$. 
Thus after $O(\log n \log( \log n/\delta))$ steps, the points will be halfed $\log n$ times and the maximum will be identified with probability $1 - \delta$. Thus, the total number of queries is $\tilde O(\log n \log(1/\delta) )$ and we obtain the following lemma.

\begin{lemma}
  \label{l:primMAX}
    There exists a procedure $\MAX(C, C_f)$ that takes as input the description $C$ of a set $S$ and a function $f: \Domain \rightarrow [M]$ given as a circuit $C_f$ and computes the value $\max_{x \in \Set} f(x)$ using $\tilde O( \log n \log( 1 /\delta) )$ conditional samples in time $\tilde O( ( |C| + |C_f| ) \log n \log( 1 /\delta) )$ with probability of error $\delta$.
\end{lemma}

\makeatletter{}\subsection{Sum of Weights of Points}
\label{s:primSUM}
  
  The \textit{sum of weights of points} function takes as input a description $C$ of a set $S$ together with a function $f : \Domain \rightarrow [M]$. The output of the function is an $(1 \pm \eps)$ approximation of the sum of all $f(x)$ for every $x$ in $\Set  = S \cap \Support$, i.e.
$\sum_{x \in \Set} f(x)$. We call this function $\SUM(C, C_f)$. \\

To implement this function in the conditional sampling model, we first compute $\MAX = \MAX(C,C_f)$ (Lemma \ref{l:primMAX}). We then create $k = \log_{1+\eps}( n / \eps) = O( \log n / \eps )$ sets $S_i = \{x \in S: f(x) \in ((1+\eps)^{-i}, (1+\eps)^{1-i} ] \cdot \MAX \}$ for $i \in [k]$, by grouping together points whose values are close. Let $C_{S_i}$ denote the circuit description of every set $S_i$.
We can get an estimate for the overall sum as  
\[ \SUM(C, C_f) = \sum_{i = 1}^{k} \SE\left( C_{S_i} \right) \frac{\MAX}{(1 + \eps)^{i - 1}} \]

To see why this is an accurate estimate we rewrite the summation in the following form:
\begin{equation}\label{eq:sum} 
  \sum_{x \in \Set} f(x) = \sum_{ \substack{x \in \Set: \\ f(x) \le \MAX \cdot (1+\eps)^{-k} }} f(x) + \sum_{i = 1}^{k} \sum_{x \in S_i \cap \Support} f(x)
\end{equation}

To bound the error for the second term of~\eqref{eq:sum}, notice that for every $i \in [k]$ and $x \in S_i$, we have that  
$f(x) \in ((1+\eps)^{-i}, (1+\eps)^{1-i} ] \cdot \MAX$. Thus, the value $|S_i \cap \Support| \frac{\MAX}{(1 + \eps)^{i - 1}}$ is a $(1+\eps)$-approximation to the sum $\sum_{x \in S_i\cap \Support} f(x)$. Since the primitive $\SE\left( C_{S_i} \right)$ returns a $(1+\eps)$-approximation to $|S_i \cap \Support|$, we get that the second term of~\eqref{eq:sum} is approximated by $\SUM(C, C_f)$ multiplicatively within $(1+\eps)^2 \le 1+3\eps$.

The first term introduces an additive error of at most $n \cdot \MAX \cdot (1+\eps)^{-k}  = \eps \cdot \MAX \le \eps \cdot \SUM(C, C_f)$ which implies that  $\SUM(C, C_f)$ gives $(1 \pm 4\eps)$-multiplicative approximation to the sum of weights. Rescaling $\eps$ by a constant, we get the desired guarantee.
Thus, we can get the estimate by using one query to the $\MAX$ primitive and $k = O( \log n / \eps )$ queries to $\SE$. For the process to succeed with probability $\delta$, we require that all $k$ of the $\SE$ queries to succeed with probability $\delta' = \delta / k$. Plugging in the corresponding guarantees of Lemmas~\ref{l:primSE} and~\ref{l:primMAX}, we obtain the following:

\begin{lemma}
  \label{l:primSUM}
    There exists a procedure $\SUM(C, C_f)$ that takes as input the description $C$ of a set $S$ and a function $f : \Domain \rightarrow [M]$ given by a circuit $C_f$ and computes an $(\eps, \delta)$-multiplicative approximation of the value $\sum_{x \in \Set} f(x)$ using 
    $\tilde O( \log n \log ( 1 / \delta ) / \eps^3 )$ conditional samples in time 
    $\tilde O\left( (|C| + |C_f| + \log^2 |\Domain|) \log n \log(1 / \delta) / \eps^3 \right)$.
\end{lemma}
 
\makeatletter{}\subsection{Weighted Sampling}
\label{s:primWCOND}

  The \textit{weighted sampling} function gets as input a description $C$ of a set $S$ together with a function $f : \Domain \rightarrow [M]$ given as a circuit $C_f$. The output of the function is a point $x$ in the set $\Set = S \cap \Support$ chosen with probability proportionally to 
the value $f(x)$. Therefore, we are interested in creating an oracle $\WCOND(C, C_f)$ that outputs element $x \in \Set$ with probability
$ \frac{f(x)}{\sum_{y \in \Set} f(y)} $.

    To implement the weighted sampling in the conditional sampling model, we use a similar idea as in support estimation. First we compute $\SUM = \SUM(C, C_f)$ and then we define a random set $H$ that contains independently every element $x$ with probability
\begin{equation}
  \label{eq:primWcond1}
  \Prob[ x \in H ] = \frac{f(x)}{2 \SUM}
\end{equation}

  Let $C_H$ be the description of $H$. We will later use Corollary \ref{cor:prsc} in order to build a pseudorandom set $H'$ with small circuit description $C_{H'}$ that approximately achieves the guarantees of $C_H$.
  
Based on the random set $H$, we describe Algorithm~\ref{alg:Wcond} which performs weighted sampling according to the function $f$.

\begin{algorithm}
  \caption{\label{alg:Wcond} Sampling elements according to their weight.}
  \begin{algorithmic}[1]
    \State $selected \leftarrow \bot$
    \While{$selected = \bot$ \textbf{and} \#iterations $ \le k$ }
      \State Construct the random set $H$ and $C_H$ as described by the equation (\ref{eq:primWcond1})
      \State Check if there exists a unique point $x \in \Set$ in the set $H$.
      \If{such unique point $x$ exists}
        \State With probability $1 - \frac{f(x)}{2 \SUM}$, set $selected \leftarrow x$
      \EndIf
    \EndWhile \\
    \Return{$selected$}
  \end{algorithmic}
\end{algorithm}

  We argue the correctness of this algorithm. Given a purely random $H$, we first show that at every iteration, the probability of selecting each point $x \in \Set$ is proportional to its weight. This implies that the same will be true for the final distribution as we perform rejection sampling on $\bot$ outcomes.

  The probability that in one iteration the algorithm will return the point $x \in S$ is the probability that $x$ has been chosen in $H$ and that $|H \cap \Set| = 1$, i.e. it is the unique point of the input set $\Support$ that lies in set $S$ and was not filtered by $H$. For every $x \in \Set$, this probability is equal to
\[ \Prob[x \in H] \prod_{y \in \Set, y \neq x} \Prob[y \notin H] \cdot \Prob[\text{keep } x] = \]
\[ = \frac{f(x)}{2 \SUM} \prod_{y \in \Set, y \neq x} \left(1 - \frac{f(y)}{2 \SUM}\right) \cdot \left(1 - \frac{f(x)}{2 \SUM}\right) = \frac{f(x)}{2 \SUM} \prod_{y \in \Set} \left(1 - \frac{f(y)}{2 \SUM}\right) \]
and it is easy to see that this probability is proportional to $f(x)$ as all other terms don't depend on $x$.

We now bound the probability of \emph{selecting} a point at one iteration. This is equal to
\[ \frac{\sum_{x \in \Set} f(x)}{2 \SUM} \prod_{y \in \Set} \left(1 - \frac{f(y)}{2 \SUM}\right) \ge \frac{1}{2} \frac{1}{1 + \eps} \exp\left(- 2 \frac{\sum_{y \in S} f(y)}{2 \SUM}\right) \ge \frac{1}{2} \frac{1}{1 + \eps} \exp\left(- \frac{1}{1 - \eps}\right) \]
which is at least $1/4$ for a small enough parameter $\eps > 0$ chosen in our estimation $\SUM$ of the total sum of $f(x)$. Thus at every iteration, there is a constant probability of outputing a point. By repeating $k = \Theta(\log( 1/\delta))$ times we get that the algorithm outputs a point with probability $\delta/2$.

Summarizing, if we assume a purely random set $H$, the probability that the above procedure will fail in $O(\log( 1/\delta))$ is at most $\delta/2$ plus the probability that the 
computation of the sum will fail which we can also make to be at most $\delta/2$, for a total probability $\delta$ of failure. Since we need only a constant multiplicative approximation to the sum, using Lemma \ref{l:primSUM} the 
total number of queries that we need for the probability of failure to be at most $\delta/2$ is $\tilde O( \log n \log ( 1 / \delta ) )$. 

Since the random set $H$ can have very large description complexity, we use Corollary~\ref{cor:prsc} to generate a pseudorandom set $H'$. If we apply the corollary for error $\delta'$ we get that the total variation distance between 
the output distribution in one step when using $H'$ with the distribution when using $H$ is at most:
\[ \sum_{x \in \Set} \abs{\Prob(\Cond(C \wedge C_H) = x) - \Prob(\Cond(C \wedge C_{H'}) = x)} \le n \delta' \]
Since we make at most $k = \Theta(\log( 1/\delta))$ queries the oracle $\Cond$, we get that the total variation distance between the two output distributions is $O( n \delta' \log( 1/\delta) )$. Setting $\delta' = O( \frac{\eps} {|\Domain| \log( 1/\delta)} )$ we get that this distance is at most $\eps$. Computing the total runtime and number of samples, we obtain the following lemma.

\begin{lemma}
  \label{l:primWCOND}
    There exists a procedure $\WCOND(C, C_f)$ that takes as input the description $C$ of a set $S$ and a function $f$ given by a circuit $C_f$ and returns a point $x \in \Set$ from a probability distribution that is at most 
  $\eps$-far in total variation distance from the probability distribution that selects each element $x \in \Set$ proportionally to $f(x)$. The procedure fails with probability at most $\delta$, uses $\tilde O( \log n \log ( 1 / \delta ) )$ conditional samples and takes time 
  $\tilde O( (|C| + |C_f| + \log^2{|\Domain|}) \log n \log ( 1 / \delta ) + \log{|\Domain|} \log{(1/\eps)} \log ( 1 / \delta ) )$.
\end{lemma}

\makeatletter{}\subsection{Distinct Elements Sampling -- $\ell_0$ Sampling}
\label{s:primDES}
  
The \textit{distinct elements sampling} function gets as input a description $C$ of a set $S$ together with a function $f : \Domain \rightarrow [M]$ described by a circuit $C_f$. It outputs samples from a distribution on the set $\Set = S \cap \Support$ 
 such that the distribution of values $f(x)$ is uniform over the image space $f(\Set)$. We thus want that for every $y \in f(\Set)$, 
$ \Prob[x \in f^{-1}(y)] = {|f(\Set)|}^{-1} $.

We first explain the implementation of the algorithm assuming access to true randomness. Assume therefore that we have a circuit $C_h$ that describes one purely random hash function 
  $h : [M] \rightarrow [M]$. Then $\argmax_{x \in \Set} h( f(x) )$ will produce a uniformly random element as long as the maximum element is unique. This means that if we call the procedure $\ARGMAX$ to find a point $x^* = \ARGMAX(C, C_h \circ C_f)$ and check that no point $x \in \Set$ exists such that $f(x) \neq f(x^*)$ and $h(f(x)) = h(f(x^*))$ then the result will be a point distributed with the correct distribution. Repeating $O(\log(1/\delta))$ times guarantees that we get a valid point with probability at least $1-\delta$.

Therefore the only question is how to replace $h$ with a pseudorandom $h'$. We can apply Nisan's pseudorandom generator. Consider an algorithm that for every point $y \in [M]$ in order, draws a random sample $s$ uniformly at random from $[M]$ and checks if $y \in f(\Set)$ and whether $s$ is the largest value seen so far. This algorithm computes $\argmax_{y \in P} h( y )$ while only keeping track of the largest sample $s$ and the largest point $y$. This algorithm uses $\Theta(\log M)$ bits of memory and $O(M \log M)$ random bits. Therefore we can apply Nisan's theorem (Theorem \ref{th:nisan}) for space $\Theta(\log(1/\eps))$ for $\eps > M^{-1}$ and we can replace $h$ with $h'$ that uses only $\tilde O(\log M \log(1/\eps) )$ random bits whose circuit representation is only $\tilde O( \log M \log(1/\eps) )$.
    
  This means that we can use the Lemma \ref{l:primMAX} and Theorem \ref{th:nisan} to get the following lemma about the $\ell_0$-sampling.

  \begin{lemma}
  \label{l:primDES}
      There exists a procedure $\DES(C, C_f)$ that takes as input the description $C$ of a set $S$ and a function $f : \Domain \rightarrow [M]$ given by a circuit $C_f$ and returns a point $x \in \Set$ from a probability distribution that is at
    most $\eps$-far in total variation distance (for $\eps \le M^{-1}$) from a probability distribution that assigns probability $\frac{1}{|f(\Set)|}$ to every set $f^{-1}(y)$ for $y \in f(\Set)$. This procedure fails with probability at most $\delta$, uses $O( \log n \log ( 1 / \delta ) )$ conditional samples and takes time $O((|C| + |C_f| + \log M \log(1/\eps)) \log n \log ( 1 / \delta ) )$.
  \end{lemma}

  \makeatletter{}\section{$k$-means clustering}
\label{s:kmeans}

   In this section we describe how known algorithms for the $k$-means clustering can be transformed to sublinear time algorithms in the case that we have access to conditional samples. The basic tool of
this algorithms was introduced by Arthur and Vassilvitskii \cite{arthurV07}.
  \begin{description}
    \item[$D^2$-sampling: ] This technique provides some very simple algorithms that can easily get a constant factor approximation to the optimal $k$-means clustering, like in the work of 
    Aggarwal et. al. \cite{aggarwalDK09}. Also if we allow exponential running time dependence on $k$ then we can also get a PTAS like in the work of Jaiswal et. al. \cite{jaiswal0S14}. The drawback of 
    this PTAS is that it works only for points in the $d$ dimensional euclidean space.
  \end{description}

\noindent When working in arbitrary metric spaces, inspired by Aggarwal et. al. \cite{aggarwalDK09} we use adaptive sampling to get a constant factor approximation to the $k$-means problem. 
Now we describe how all these steps can be implemented in sublinear time using the primitives that we described in the previous section.
The steps of the algorithm are:

\begin{enumerate}
  \item \textbf{Pick a set $P$ of $O(k)$ points according to $D^2$-sampling.} \\
    For $O(k)$ steps, let $P_i$ denote the set of samples that we have chosen in the $i$ step. We pick the $(i + 1)$th point according to the following distribution
    \[ \mathbb{P}(\text{probability of picking  } \vec{x}_{k}) = \frac{d^2(\vec{x}_k, P_i)}{\sum_j d^2(\vec{x}_j, P_i)} \]

    \noindent \textbf{Implementation:} To implement this step we simply use the primitive $\WCOND(C, f)$ where $C$ is the constant true circuit and $f(x) = d^2(x, P_i) = \min_{p \in P_i} d^2(x, p)$. The circuit to implement the function
    $d^2( \cdot, \cdot )$ has size $\tilde{O}(\log \abs{\Domain})$ where $\Delta$ is the diameter of the space. Now since $|C_i| \le O(k)$ we can also implement the minimum using a tournament with only $O(k)$ 
    comparisons each of which has size $O(\log \abs{\Domain})$. This means that the size of the circuit of $f$ is bounded by $|C_f| \le \tilde{O}(k \log \abs{\Domain})$. Therefore we can now use the Lemma \ref{l:primWCOND}
    and get that we need $\tilde{O}(k \log n \log ( 1 / \delta ))$ queries and running time $\tilde{O}((k \log \abs{\Domain} + \log^2 \abs{\Domain}) \log n \log (1 / \delta) + \log \abs{\Domain} \log (1/\eps) \log (1 / \delta))$ 
    to get the $O(k)$ needed samples from a distribution that is $\eps_1$ in total variation distance from the correct distribution and has probability of error $\delta$ for each sample.
  \item \textbf{Weight the points of $P$ according to the number of points that each one represents.} \\
    For any point $p \in P$ we set
    \[ w_p = \abs{\left\{ x \in \Omega \mid \forall p'( \neq p) \in P ~~ d(x, p') < d(x, p) \right\}} \]

    \noindent \textbf{Implementation:} To implement this step given the previous one we just iterate over all the points in $P$ and for each one of these points $p$ we compute the weight $w_p$ using the procedure $\SUM$ as described in
    Lemma \ref{l:primSUM}. Similarly to the previous step we have that $C$ is the constant 1 circuit and $f_p(x)$ is equal to 1 if the closest point to $x$ in $P$ is $p$ and zero otherwise. To 
    describe this function we need as before $O(k \log \abs{ \Domain })$ sized circuit. Therefore for this step we need $\tilde{O}(\log n \log(1 / \delta) / \eps_2^3)$ conditional samples and running time
    $\tilde{O}((k \log \abs{ \Domain } + \log^2 \abs{\Domain}) \log n \log (1 / \delta) / \eps_2^3)$ in order to get an $(\eps_2, \delta)-$multiplicative approximation of every $w_p$.
  \item \textbf{Solve the weighted $k$-means problem in with the weighted points of P.} \\
    This can be done using an arbitrary constant factor approximation algorithm for $k$-means since the size of $P$ is $O(k)$ and therefore the running time will be just $\poly(k)$ which is already sublinear in
    $n$.
\end{enumerate}

\noindent To prove that this algorithm gets a constant factor approximation we use Theorem 1 and Theorem 4 of \cite{aggarwalDK09}. From the Theorem 1 of \cite{aggarwalDK09} and the fact that we sample from a distribution that is 
  $\eps$-close in total variation distance to the correct one we conclude that the set $P$ that we chose satisfies Theorem 1 of \cite{aggarwalDK09} with probability of error at most $\eps_1 + O(k) \delta$. Then it is easy to 
  see at Theorem 4 of \cite{aggarwalDK09} that when we have a constant factor approximation of the weights, we lose only a constant factor in the approximation ratio. Therefore we can choose $\eps_2$ to be constant. Finally 
  for the total probability of error to be constant, we have to pick $\eps_1$ to be constant. Combining all these with the description of the step that we have above we get the following result. 

\begin{theorem}
  \label{th:constkmeans}
    There exists an algorithm that computes an $O(1)$-approximation to the $k$-means clustering and uses only $\tilde{O}(k^2 \log n \log ( k / \delta ))$ conditional queries and has running time \\
    $\tilde{O}(\poly(k) \log^2 \abs{\Domain} \log n \log ( 1 / \delta ))$.
\end{theorem}

\noindent \textbf{Remark 1.} The above algorithm could be extended to an arbitrary metric space where we are given a circuit $C_d$ that describes the distance metric function. In this case the running time will also depend on $\abs{C_d}$. \\

\noindent \textbf{Remark 2.} In this case that the points belong to $d$ dimensional space, we can also use the \textbf{Find-k-means($\Support$)} algorithm by \cite{jaiswal0S14} to get $(1 + \eps)$-approximation instead of constant. 
This algorithm iterates over a number of subsets of the input of specific size that have been selected using $D^2$-sampling. Then from all these different solution it selects the one with the minimum cost. We can implement this 
algorithm using our $\WCOND$ and $\SUM$ primitives to get a sublinear $(\eps, \delta)$-multiplicative approximation algorithm that uses 
$\tilde{O}\left(2^{\tilde{O}(k^2 / \eps)} \cdot \log n \cdot \log (1/\delta)\right)$ conditional samples and has running time
$\tilde{O}\left(2^{\tilde{O}(k^2 / \eps)} \cdot \log^2 \abs{ \Domain } \cdot \log n \cdot \log (1/\delta)\right)$.

  \makeatletter{}\section{Euclidean Minimum Spanning Tree}
\label{s:mst}

In this section we are going to discuss how to use the primitives that we described earlier in order to estimate the weight of the minimum spanning tree of $n$ points in euclidean space. 

More specifically, suppose that we have a $d$-dimensional discrete euclidean space $\Domain = \{1,\dots , \Delta \}^d$ and a set of $n$ points $\Support =\{x_1,\dots ,x_n\}$, where $x_i \in \Domain$. We assume that $\Delta=poly(n)$ which is a reasonable assumption to make when bounded precision arithmetic is used. This means that each coordinate of our points can be specified using $O(\log n)$ bits. 

In what follows, we are going to be using the following formula that relates the weight of an MST to the number of connected components of certain graphs. Let $W$ denote the maximum distance between any pair of points in $\Domain$. Moreover, let $G^{(i)}$ be the graph whose vertices correspond to points in $\Support$ and the is an edge between two vertices if and only if the distance of the corresponding points is at most $(1+\varepsilon)^i$. By $c_i$ we denote the number of connected components of the graph $G^{(i)}$. In \cite{CS04}, it is shown that the following quantity leads to an $(1+\varepsilon)$-multiplicative approximation of the weight of the minimum spanning tree:    

\begin{equation}\label{MST_approx}
n-W + \varepsilon\cdot\sum_{i=0}^{\log_{1+\varepsilon}W-1} (1+\varepsilon)^i \cdot c_i
\end{equation}
The quantity would be equal to the weight of the minimum spanning tree if all pairwise distances between points were $(1+\eps)^i$ for some $i\in \mathbb{N}$. 

In order to estimate the weight of the MST, we need to estimate the number of connected components $c_i$ for each graph $G^{(i)}$. As shown in~\cite{FIS08}, for every $i$, we can equivalently focus on performing the estimation task after rounding the coordinates of all points to an arbitrary grid of size $\eps (1+\eps)^i / \sqrt{d}$. This introduces a multiplicative error of at most $1+O(\eps)$ which we can ignore by scaling $\eps$ by a constant.

We thus assume that every point is at a center of a grid cell when performing our estimation. We perform a sampling process which 
samples uniformly from the occupied grid cells (regardless of the number of points in each of them) and estimates the number of cells covered by the connected component $j$ of $G^{(i)}$ that the sampled cell belongs to. Comparing that estimate to an estimate for the total number of occupied grid cells, we obtain an estimate for the total number of connected components. In more detail, if the sampled component covers a $\rho$ fraction of the cells, the guess for the number of components is $\frac 1 \rho$. For that estimator to be accurate, we need to make sure that the total expected number of occupied grid cells is comparable to the total number of components without blowing up in size exponentially with the dimension $d$ of the Euclidean space. We achieve this by choosing a uniformly random shifted grid. This random shift helps us avoid corner cases where a very small component spans a very large number of cells even when all its contained points are very close together. With a random shift, such cases have negligible probability.

We will first use the following lemma to get an upper bound on the number of occupied cells which holds with high probability:

\begin{lemma}\label{lem:minkowski}
Let $C\subseteq \mathbb{R}^d$ be a $1$-D curve of length $L$, a grid $G\subseteq \mathbb{R}^d$ of side length $R$ and random vector $\vec{v}$ distributed uniformly over $[0,R]^d$. Then, the expected number of grid cells of $G+\vec{v}$ that contain some point of the curve is $\frac{vol([0,R]^d+C)}{R^d}$, where "+" denotes the Minkowski sum of the two sets.
\end{lemma} 

\begin{proof}
  Consider the grid $G=\{R \vec z: \vec z \in \mathbb{Z}^d\} \subseteq \mathbb{R}^d$ shifted by a random vector $\vec v \in [0,R]^d$ to obtain a grid $G_v = \vec v + G$. We associate every point $\vec z \in G_v$ with a cell $\vec z + [0,R]^d$. Observe that a cell corresponding to a grid point $\vec z$ intersects the curve $C$ if $(\vec z + [0,R]^d) \cap C \neq \emptyset$, or equivalently if $\vec z \in C + [-R,0]^d$. The expected number of occupied grid cells is thus equal to the expected number of grid points of $G_v$, which lie in the Minkowski sum of $C$ and $[-R,0]^d$.
  
Note that each of the original grid points $\vec z \in G$ can move inside a $d$-dimensional hypercube of side length $R$ and all those hypercubes are pairwise disjoint and span the whole $d$-dimensional space. 

Now let $I_{\vec z}$ be the indicator random variable for the event that $\vec z \in C+[-R,0]^d$. Clearly,
\[
\mathbb{E}[I_{\vec z}]=\Pr[\vec z \in C+[-R,0]^d]=\frac{vol((C+[-R,0]^d) \cap (\vec z +[0,R]^d))}{R^d}
\]

So the expected number of points in $C+[-R,0]^d$ is:
\[
\mathbb{E}[\sharp points]=\mathbb{E}\left[\sum_{\vec z \in G} I_{\vec z}\right]=\frac{vol(C+[-R,0]^d)}{R^d}  =\frac{vol(C+[0,R]^d)}{R^d}
\]
  
\end{proof}

  The following lemma bounds the volume of the required Minkowski sum.

\begin{lemma}\label{lem:minkowskibound}
Let $C\subseteq \mathbb{R}^d$ be a $1$-D curve of length $L$. The volume of the Minkowski sum $C+[0,R]^d$ is at most $R^d+\sqrt{d}\cdot L$.
\end{lemma}

\begin{proof}
We can think of the Minkowski sum as the set of points spanned  by the $d$-dimensional hypercube $[0,R]^d$ as it travels along the curve $C$. Now suppose that we move the hypercube for a very small distance $dL$ along an arbitrary unit vector $\vec{r}=(a_1,\dots, a_d)$ with positive coordinates (we assume this wlog since all other cases are symmetric). Also, let $e_1,\dots ,e_d$ be the standard basis vectors (i.e $e_i=(0,\dots,0,1,0\dots,0)$ where the $"1"$ is at the $i$-th coordinate). Note that each of those vectors is orthogonal to a facet of the hypercube and the total volume spanned by each facet during the movement is equal to the absolute value of the inner product $\vec{r}\cdot\vec{e_i}$ scaled by $dL$ where $e_i$ is the standard basis vector orthogonal to that facet.   

The volume spanned by this displacement is equal to the sum of the volumes spanned by each of the facets and is given by the following formula:
\[
dV= dL\cdot R^{d-1} \cdot \sum_{i=1}^d  \vec{r}\cdot \vec{e_i} =   dL\cdot R^{d-1}\cdot \sum_{i=1}^d  \vert a_i \vert = dL \cdot R^{d-1}\cdot \Vert r \Vert_1
\] 
  
So in the worst case the curve $C$ is a straight line segment along the all ones unit vector $\vec{w}=\frac{1}{\sqrt{d}}(1,1,\dots , 1)$, since this is the unit vector that has the maximum $l_1$ norm.

In this case, the total volume spanned during the movement along $C$ is:

\[
V_C= L \cdot R^{d-1}\cdot \Vert r \Vert_1 = \sqrt{d}\cdot R^{d-1}\cdot L
\] 

So, the volume of the Minkowski sum $C+[0,R]^d$ is: 
\[
vol(C+[0,R]^d)= R^d + V_C = R^d+  \sqrt{d}\cdot R^{d-1}\cdot L
\] 
and this is an upper bound for the general case. 
\end{proof}

  We can view the minimum spanning tree as a 1-D curve considering its Euler tour. The length of this Euler tour is $2\cdot MST$ since each edge is traversed exactly twice. For the same reason, each point in the Minkowski sum is "covered by" at least two points in the curve. So, effectively, the length of the curve can be divided by $2$.
Thus, the volume of the Minkowski sum $T + [0,R]^d$ is at most $R^d + \sqrt{d}\cdot R^{d-1}\cdot MST$. Therefore, by lemma \ref{lem:minkowski}, we get that:
\[
\mathbb{E}[\sharp cells]=1+\frac{\sqrt{d}\cdot MST}{R}
\]

\iffalse
Since our techniques involve the use of a randomly shifted  rectangular grid of side length $R$, we will need the following lemma which shows that the MST occupies a polynomial in $d$ number of grid cells even though this number is exponential in $d$ in the worst case. This will help us get a polynomial dependence in $d$ for both the running time and the number of conditional queries.

\begin{lemma}\label{exp_grid_cells}
Let $T\subseteq \mathbb{R}^d$ be a tree in which the sum of all edge lengths is denoted by $MST$. Also, consider a rectangular grid of side length $R$ in $\mathbb{R}^d$ on which a uniformly random grid has been applied (i.e the probability density for being point of the grid is equal at every point in $\mathbb{R}^d$). Then, the expected number of grid cells that contain some point of the tree is at most $1+\frac{\sqrt{d}\cdot MST}{R}$.
\end{lemma}
\fi

\iffalse
\begin{proof}
The length of an Euler tour of the minimum spanning tree $2\cdot MST$ since each edge is traversed exactly twice. For the same reason, each point in the Minkowski sum is "covered by" at least two points in the curve. So, effectively, the length of the curve can be divided by $2$.
Thus, the volume of the Minkowski sum $T+[0,R]^d$ is at most $R^d+ \sqrt{d}\cdot R^{d-1}\cdot MST$.  

By lemma \ref{lem:minkowskibound}, we get that:
\[
\mathbb{E}[\sharp points]=1+\frac{\sqrt{d}\cdot MST}{R}
\]
\end{proof}
\fi

Using Markov's inequality, we get that 
\[
\Pr\left[ \sharp cells >2\cdot (1+\frac{\sqrt{d}\cdot MST}{R})\right]< \frac12
\]

Finally, we can use our support estimation primitive of Lemma~\ref{l:primSE} to estimate the number of occupied grid cells after a random shift which enables us to amplify the success probability to $1-\delta$ by picking the random shift with the smallest number of cells after $O(\log (1/\delta))$ repetitions. 

We immediately get the following corollary:

\begin{corollary}
    We can find a grid of side length $R$, such that the number of grid cells that contain points is at most $2 \cdot (1+\frac{\sqrt{d}\cdot MST}{R})$ using $\tilde{O}(\log \log n \log^2 (1/\delta) )$ conditional samples while the failure probability is at most $\delta$. The total running time is $\tilde{O}(\log^2 n \log^2 (1/\delta) )$.
\end{corollary}
\subsection{Computing the size of small connected components}

As we have said earlier, we will use~\eqref{MST_approx} in order to estimate the weight of the MST. For every $i$, we estimate the number of connected components $c_i$ in the graph $G^{(i)}$ assuming that the points are in the center of a given grid with side length $R = \eps (1+\eps)^i / \sqrt{d}$ and that the total number of grid cells is at most $2 \cdot (1+\frac{\sqrt{d}\cdot MST}{ R }) = 2 \cdot (1+\frac{d \cdot MST}{ \eps (1+\eps)^i })$.
For that purpose, we will sample grid cells uniformly at random and estimate the size of the connected  that we hit during our sampling procedure. 

In order to do this, we first sample a uniformly random grid cell using the Distinct Elements Sampling procedure of Lemma~\ref{l:primDES} and then perform a BFS-like search starting from that cell to count the number of cells in that connected component. More specifically, at every iteration, we ask for a uniformly random cell that is adjacent in $G^{(i)}$ to one of the cells we have already visited, using our conditional sampling oracle. If we visit more than $t$ distinct cells (for some threshold $t$) during this search, we will stop and output that the connected component is ``big". Otherwise, we stop when we completely explore the connected component and output its size. Since there cannot be too many ``big" connected components, ignoring them is not going to affect our final estimate too much. More specifically, we will set $t=\frac{d \log_{1+\eps} W}{\varepsilon }$ and note that there can be at most $ \sharp cells / t$ ``big" connected components. 
So, by ignoring them we introduce at an additive error of at most $\sharp cells / t$.   

\subsection{Algorithm for estimating the number of connected components}

Now, we can continue with the main part of the algorithm which shows how to estimate the number of connected components of the graph $G^{(i)}$.  

Let $s_{1},\dots, s_{k}$ be the number of cells each of the $k$ connected components of $G^{(i)}$ occupies respectively and $X$ be the random variable for the index of the connected component our sample hits (in number of cells). Also, let $S$ be the total number of occupied grid cells, and $\hat{S}$ be our estimate for that using the $\SE$ primitive from Lemma~\ref{l:primSE}. 

\begin{algorithm}
  \caption{\label{alg:NumComponents} Estimating $c_i$}
  \begin{algorithmic}[1]
    \State $x_0 \leftarrow$ uniformly random occupied cell using $l_0$ sampling. 
      \State $U \leftarrow \{x_0\}$
      \State $s\leftarrow 1$
      \While {$(C_U\not\equiv 0) \wedge (s<\frac{1}{\varepsilon})$}
      \State $x \leftarrow COND(C_U)$, where the circuit $C_U$ requires the output to have a neighbor in he set $U$. 
      \State $U \leftarrow U \cup \{x\}$ 
      \State $s\leftarrow s+1$ 
      \EndWhile
      \If {$s<t=\frac{d \log_{1+\eps} W}{\varepsilon }$}
      \State{\Return{$\hat{c}_i=\frac{SE(C)}{s}$}}
      \Else 
      \State{\Return{$\hat{c}_i=1$}}
      \EndIf
  \end{algorithmic}
\end{algorithm}

Our estimator which comes from Algorithm~\ref{alg:NumComponents} is the following:
\[
\sum_{j=1}^k \mathbb{I}[X=j]\frac{\hat S}{\hat{s}_{j}}
\]  
where $\hat{s}_{j}=s_{j}$ when $s_{j}<t$ and $\hat{s}_{j} = \hat S$ otherwise (i.e $\hat{c}_i=1$).

Note that $\hat{s}_{j}$ is always overestimating the true value of $s_{j}$ and an lower bound for this expectation is the same sum if we exclude the ``big" connected components. This means that:
\[
\sum_{\substack{j=1 \\ s_{j}<t}}^k \Pr[X=j]\frac{\mathbb{E}[\hat{S}]}{s_{j}} \leq \mathbb{E}[\hat{c}_i]\leq \sum_{\substack{j=1 \\ s_{j}<t}}^k \Pr[X=j]\frac{\mathbb{E}[\hat{S}]}{s_{j}} + \left\vert \left\{j:s_{j}>t\right\} \right\vert \Leftrightarrow
\]  
\[
\left\vert \left\{j:s_{j}\leq t\right\} \right\vert \cdot \frac{\mathbb{E}[\hat{S}]}{S}\leq \mathbb{E}[\hat{c}_i] \leq \left\vert \left\{j:s_{j}\leq t\right\} \right\vert \cdot \frac{\mathbb{E}[\hat{S}]}{S} + \left\vert \left\{j:s_{j}>t\right\} \right\vert
\]  
where we substituted $Pr[X = j] = \left\vert \left\{j:s_{j}\leq t\right\} \right\vert / S$ since every component is selected with probability proportional to the number of cells it contains.

From the $\SE$ primitive, we have that 
\[
\Pr[\vert S-\hat{S} \vert > \varepsilon \cdot S]<\delta
\]

So, by conditioning on the support estimation procedure succeeding (which happens with probability $1-\delta$) we get that: 

\[
\left\vert \left\{j:s_{j}\leq t\right\} \right\vert \cdot (1-\varepsilon)
\leq \mathbb{E}[\hat{c}_i] 
\leq \left\vert \left\{j:s_{j}\leq t\right\} \right\vert \cdot (1+\varepsilon) + \left\vert \left\{j:s_{j}>t\right\} \right\vert
\]  
 
Thus, we have that $\mathbb{E}[\hat{c}_i]$, is an accurate approximation of $c_i$ with probability at least $1-\delta$.
In particular, since $\left\vert \left\{j:s_{j}>t\right\} \right\vert < S/t$ we get:
\[
(1-\eps) c_i - S/t
\leq \mathbb{E}[\hat{c}_i] \leq (1 + \eps) c_i + S/t
\]
with probability at least $1-\delta$.

We are going to repeat the above estimation $m$ times independently and keep the average which, as we are going to show, will be very well concentrated around its mean $\mu_i=\mathbb{E}[\hat{c}_i]$. To show that, we can use Hoefding's inequality since our trials are independent and trivially the value of each individual estimate is lower and upper bounded by $1$ and $S$ respectively.   

Let $\hat \mu_i$ denote the estimated average. From Hoefding's inequality we get:
\[
\Pr[\vert\hat{\mu}_i-\mu_i\vert > S/t]< e^{\frac{-2 m^2 \frac {S^2}{t^2} S^2 }{m\cdot S^2}}=e^{- 2 m / t^2 }
\]

If we set $m=O(t^2)$ we get the above guarantee with probability at least $1-\delta$. This means that with probability $1-\delta$ we get:
\[
(1-\eps) c_i - 2 S/t
\leq \hat \mu_i \leq (1 + \eps) c_i + 2 S/t
\]

This means that $\hat \mu_i = (1\pm\eps) c_i \pm 4 \cdot (1+\frac{d \cdot MST}{ \eps (1+\eps)^i }) / t$
and applying it to equation \eqref{MST_approx}, we get the following estimator for the weight of the MST:
 \[
 \widehat{MST}=n-W + \varepsilon\cdot\sum_{i=0}^{l=\log_{1+\varepsilon}W-1} (1+\varepsilon)^i \cdot \hat{\mu}_i = (1 \pm \eps) MST \pm  4 \eps \sum_{i=0}^{l=\log_{1+\varepsilon}W-1} (1+\varepsilon)^i\cdot \frac{d \cdot MST}{ t \eps (1+\eps)^i }
 \]
The last term is bounded by
$ \frac{ 4 \cdot d \cdot \log_{1+\varepsilon} W }{ t } MST $ which for $t = \frac{d \log_{1+\eps} W}{\varepsilon }$ gives an
$(1 + \eps)$-multiplicative approximation to $MST$.

The total runtime requires $\log_{1+\varepsilon} W$ iterations to estimate every $c_i$ by $\hat \mu_i$. In every iteration a random shifting is performed and the total number of occupied grid cells are counted using the $\SE$ primitive. Moreover, for the estimation $O(t^2)$ samples are required using Distinct Element Sampling ($\ell_0$-sampling) of the occupied grid cells. Finally for each such sample, a BFS procedure is performed for at most $t$ iterations. The circuit complexity of the conditional sampling queries that are required is negligible in most cases as it is subsumed by the runtime of the corresponding algorithmic primitive used. Only queries performed during the BFS have large circuit size as the circuit requires to keep all grid cells that have been visited. The size in that case is bounded by $O(t \log |\Domain|) = O(t \cdot d \cdot \log n)$. The number of samples is bounded by $\tilde O(\log_{1+\varepsilon} W t^3 ) = \tilde O(d^3 \log^4 n / \eps^7 )$ and the total runtime is bounded by $\tilde O( \log_{1+\varepsilon} W t^4 d \log n) = \tilde O(d^5 \log^6 n / \eps^9)$ if we require constant success probability. Repeating $\log(1/\delta)$ times we can amplify the total success probability to $1-\delta$.
The following theorem shows the dependence of the running time and query complexity on the parameters $n,d,\varepsilon,\delta$: 
\begin{theorem}
It is possible to compute an $(\eps,\delta)$-multiplicative approximation to the weight of the Euclidean minimum spanning tree using $\tilde O(d^3 \log^4 n / \eps^7 )\cdot \log(1/\delta)$ conditional queries in time 
$\tilde O(d^5 \log^6 n / \eps^9)\cdot \log(1/\delta)$.
\end{theorem}

  \makeatletter{} 

  \makeatletter{} 

\section{Conclusions and Future Directions}
  In this work we introduced a computational model based on conditional sampling and showed how various combinatorial optimization tasks can be performed efficiently with very limited access to the input through the conditional 
oracle. This provides a generic way to design algorithms for several other models such as parallel computation, streaming and distributed.

  In terms of future research it is interesting to explore what other tasks are approachable using this computational model and to understand its powers and limitations. A more concrete question is whether we can avoid the dependence
of the running time on the domain size under a slight variation of our model where the description of the sets are given by arithmetic circuits. This dependence is necessary in our model since specifying a single point in the input
uses at least $\log |\Domain|$ bits.

	\bibliographystyle{alpha}
	\bibliography{ref}

\newcommand{\etalchar}[1]{$^{#1}$}
\begin{thebibliography}{CFGM13}

\bibitem[ABS{\etalchar{+}}12]{afratiBSHSU12}
Foto~N. Afrati, Magdalena Balazinska, Anish~Das Sarma, Bill Howe, Semih
  Salihoglu, and Jeffrey~D. Ullman.
\newblock Designing good algorithms for mapreduce and beyond.
\newblock In {\em {ACM} Symposium on Cloud Computing, {SOCC} '12, San Jose, CA,
  USA, October 14-17, 2012}, page~26, 2012.

\bibitem[ACK15a]{acharyaCK15b}
Jayadev Acharya, Cl{\'{e}}ment~L. Canonne, and Gautam Kamath.
\newblock Adaptive estimation in weighted group testing.
\newblock In {\em {IEEE} International Symposium on Information Theory, {ISIT}
  2015, Hong Kong, China, June 14-19, 2015}, pages 2116--2120, 2015.

\bibitem[ACK15b]{acharyaCK15}
Jayadev Acharya, Cl{\'{e}}ment~L. Canonne, and Gautam Kamath.
\newblock A chasm between identity and equivalence testing with conditional
  queries.
\newblock In {\em Approximation, Randomization, and Combinatorial Optimization.
  Algorithms and Techniques, {APPROX/RANDOM} 2015, August 24-26, 2015,
  Princeton, NJ, {USA}}, pages 449--466, 2015.

\bibitem[ADK09]{aggarwalDK09}
Ankit Aggarwal, Amit Deshpande, and Ravi Kannan.
\newblock Adaptive sampling for k-means clustering.
\newblock In {\em Approximation, Randomization, and Combinatorial Optimization.
  Algorithms and Techniques, 12th International Workshop, {APPROX} 2009, and
  13th International Workshop, {RANDOM} 2009, Berkeley, CA, USA, August 21-23,
  2009. Proceedings}, pages 15--28, 2009.

\bibitem[AV07]{arthurV07}
David Arthur and Sergei Vassilvitskii.
\newblock k-means++: the advantages of careful seeding.
\newblock In {\em Proceedings of the Eighteenth Annual {ACM-SIAM} Symposium on
  Discrete Algorithms, {SODA} 2007, New Orleans, Louisiana, USA, January 7-9,
  2007}, pages 1027--1035, 2007.

\bibitem[CEF{\etalchar{+}}05]{CEFMNRS05}
Artur Czumaj, Funda Erg{\"{u}}n, Lance Fortnow, Avner Magen, Ilan Newman,
  Ronitt Rubinfeld, and Christian Sohler.
\newblock Approximating the weight of the euclidean minimum spanning tree in
  sublinear time.
\newblock {\em {SIAM} J. Comput.}, 35(1):91--109, 2005.

\bibitem[CFGM13]{chakraborty2013power}
Sourav Chakraborty, Eldar Fischer, Yonatan Goldhirsh, and Arie Matsliah.
\newblock On the power of conditional samples in distribution testing.
\newblock In {\em Proceedings of the 4th conference on Innovations in
  Theoretical Computer Science}, pages 561--580. ACM, 2013.

\bibitem[CRS14]{canonne2014testing}
Cl{\'e}ment Canonne, Dana Ron, and Rocco~A Servedio.
\newblock Testing equivalence between distributions using conditional samples.
\newblock In {\em Proceedings of the Twenty-Fifth Annual ACM-SIAM Symposium on
  Discrete Algorithms}, pages 1174--1192. Society for Industrial and Applied
  Mathematics, 2014.

\bibitem[CRT05]{CRT05}
Bernard Chazelle, Ronitt Rubinfeld, and Luca Trevisan.
\newblock Approximating the minimum spanning tree weight in sublinear time.
\newblock {\em SIAM Journal on Computing}, 34(6):1370--1379, 2005.

\bibitem[CS04]{CS04}
Artur Czumaj and Christian Sohler.
\newblock Estimating the weight of metric minimum spanning trees in
  sublinear-time.
\newblock In {\em Proceedings of the Thirty-sixth Annual ACM Symposium on
  Theory of Computing}, STOC '04, pages 175--183, New York, NY, USA, 2004. ACM.

\bibitem[CS07]{czumajS07}
Artur Czumaj and Christian Sohler.
\newblock Sublinear-time approximation algorithms for clustering via random
  sampling.
\newblock {\em Random Struct. Algorithms}, 30(1-2):226--256, 2007.

\bibitem[FIS08]{FIS08}
Gereon Frahling, Piotr Indyk, and Christian Sohler.
\newblock Sampling in dynamic data streams and applications.
\newblock {\em Int. J. Comput. Geometry Appl.}, 18(1/2):3--28, 2008.

\bibitem[Ind99]{indyk99}
Piotr Indyk.
\newblock Sublinear time algorithms for metric space problems.
\newblock In {\em Proceedings of the Thirty-first Annual ACM Symposium on
  Theory of Computing}, STOC '99, pages 428--434, New York, NY, USA, 1999. ACM.

\bibitem[JKS14]{jaiswal0S14}
Ragesh Jaiswal, Amit Kumar, and Sandeep Sen.
\newblock A simple {D} 2-sampling based {PTAS} for k-means and other clustering
  problems.
\newblock {\em Algorithmica}, 70(1):22--46, 2014.

\bibitem[MOP01]{mishraOP01}
Nina Mishra, Daniel Oblinger, and Leonard Pitt.
\newblock Sublinear time approximate clustering.
\newblock In {\em Proceedings of the Twelfth Annual Symposium on Discrete
  Algorithms, January 7-9, 2001, Washington, DC, {USA.}}, pages 439--447, 2001.

\bibitem[MOP04]{meyersonOP04}
Adam Meyerson, Liadan O'Callaghan, and Serge~A. Plotkin.
\newblock A \emph{k}-median algorithm with running time independent of data
  size.
\newblock {\em Machine Learning}, 56(1-3):61--87, 2004.

\bibitem[MP04]{mettuP04}
Ramgopal~R. Mettu and C.~Greg Plaxton.
\newblock Optimal time bounds for approximate clustering.
\newblock {\em Machine Learning}, 56(1-3):35--60, 2004.

\bibitem[Nis90]{Nisan90}
N.~Nisan.
\newblock Pseudorandom generators for space-bounded computations.
\newblock In {\em Proceedings of the Twenty-second Annual ACM Symposium on
  Theory of Computing}, STOC '90, pages 204--212, New York, NY, USA, 1990. ACM.

\bibitem[VV11]{valiantV11}
Gregory Valiant and Paul Valiant.
\newblock Estimating the unseen: an n/log(n)-sample estimator for entropy and
  support size, shown optimal via new clts.
\newblock In {\em Proceedings of the 43rd {ACM} Symposium on Theory of
  Computing, {STOC} 2011, San Jose, CA, USA, 6-8 June 2011}, pages 685--694,
  2011.

\end{thebibliography}

\end{document}